\documentclass{article}
\usepackage{amsmath,amssymb,theorem,enumerate}
\usepackage[dvips]{graphicx,color}
\usepackage[margin=1in]{geometry}
\usepackage{cite}

\makeatletter
\newenvironment{proof}[1][\proofname]{\par\normalfont
  \topsep6pt plus6pt\trivlist\item[\hskip\labelsep\itshape
  #1\@addpunct{:}]\ignorespaces}{\qed\endtrivlist}
\newcommand{\proofname}{Proof}
\DeclareRobustCommand{\qed}{%
  \ifmmode
  \else\leavevmode\unskip\penalty9999\hbox{}\nobreak\hfill\fi
  \quad\hbox{\qedsymbol}}
\newcommand{\qedsymbol}{\openbox}
\newcommand{\openbox}{\leavevmode\hbox to.77778em{%
    \hfil\vrule\vbox to.675em{%
      \hrule width.6em\vfil\hrule}\vrule\hfil}}
\makeatother

\newcommand{\Prb}{\mathsf{P}}\newcommand{\Exp}{\mathsf{E}}

\newcommand{\N}{\mathbb{N}}
\newcommand{\R}{\mathbb{R}}
\newcommand{\dd}{\mathrm{d}}\newcommand{\ee}{\mathrm{e}}

\newcommand{\SIR}{\mathsf{SIR}}

\newcommand{\bsym}[1]{\boldsymbol{#1}}
\newcommand{\ind}[1]{\boldsymbol{1}_{#1}}

\let\Bar\overline

\newcommand{\TPP}{(\mathrm{TPP})}
\newcommand{\MCP}{(\mathrm{MCP})}

\theorembodyfont{\itshape}

\newtheorem{theorem}{Theorem}

\newtheorem{lemma}{Lemma}
\theorembodyfont{\rmfamily}
\newtheorem{example}{Example}

\title{Downlink coverage probability in cellular networks with
  Poisson-Poisson cluster deployed base stations}
\author{Naoto Miyoshi\\
  Tokyo Institute of Technology}
\date{}

\begin{document}\sloppy\allowbreak\allowdisplaybreaks
\maketitle

\begin{abstract}
Poisson-Poisson cluster processes~(PPCPs) are a class of point
processes exhibiting attractive point patterns.
Recently, PPCPs are actively studied for modeling and analysis of
heterogeneous cellular networks or device-to-device networks.
However, surprisingly, to the best knowledge of the author, there is
no exact derivation of downlink coverage probability in a numerically
computable form for a cellular network with base stations (BSs)
deployed according to a PPCP within the most fundamental setup such as
single-tier, Rayleigh fading and nearest BS association.
In this paper, we consider this fundamental model and derive a
numerically computable form of coverage probability.
To validate the analysis, we compare the results of numerical
computations with those by Monte Carlo simulations and confirm the
good agreement.
\end{abstract}

\section{Introduction}

Poisson-Poisson cluster processes~(PPCPs) are a class of point
processes~(PPs) exhibiting attractive (clustering) point
patterns~(see, e.g., \cite{BlasYoge14}).
A stationary PPCP is constructed by independent, identical and finite
Poisson point processes~(PPPs), called daughter processes, placed
around points of a homogeneous PPP, called a parent process (detailed
in the next section).
Recently, PPCPs are actively studied for modeling and analysis of
heterogeneous cellular networks~(HetNets) or device-to-device~(D2D)
networks (see, e.g.,
\cite{SuryMollFett15,DengZhouHaen15,ChunHasnGhra15,SahaAfshDhil17,AfshDhil18,AfshDhilChon16,YiLiuNall17,JoshMall18}).
This is because locations of small (pico or femto) base stations~(BSs)
in HetNets or user devices in D2D networks are distributed in a
clustering nature in user hotspots.
However, surprisingly, to the best knowledge of the author, there is
no exact derivation of downlink coverage probability in a numerically
computable form for a cellular network with BSs deployed according to a
PPCP within the most fundamental setup such as single-tier, Rayleigh
fading and nearest BS association.
In this paper, we challenge this fundamental problem.

Indeed, there are several related results.
Suryaprakash~\textit{et al.}~\cite{SuryMollFett15} and Deng~\text{et
  al.}~\cite{DengZhouHaen15} study two-tier HetNets, where macro BSs
are deployed according to a homogeneous PPP and small BSs are
according to a PPCP.
Both of \cite{SuryMollFett15} and \cite{DengZhouHaen15} derive the
Laplace transform of downlink interference and, using this,
conditional downlink coverage probability given the distance to the
serving BS.
One may think that our fundamental problem is covered by their results
combined with the distribution of contact distance (distance to the
nearest point from the origin) of PPCPs derived
in~\cite{AfshSahaDhil17,AfshSahaDhil17b} (as suggested in
\cite{BlasHaenKeelMukh18}).
However, the problem is not so optimistic because we have to take into
account the correlation between the locations of the serving BS and
the interferers through the sharing parent point (PPCPs do not have
the property of independent increments unlike PPPs).
Chun~\textit{et al.}~\cite{ChunHasnGhra15} consider a $K$-tier
downlink HetNet, where BSs in each tier are deployed according to a
PPCP.
They, however, assume orthogonal multiple access and do not consider
interference from BSs in the same cluster as the serving BS.
%
Saha~\textit{et al.}~\cite{SahaAfshDhil17} extensively investigate
several models of HetNets using PPPs and PPCPs.
Though their models cover one of the most fundamental settings as a
special case, a difference from ours is that they consider the
max-SIR association, where a user is associated with the BS offering
the maximum signal-to-interference ratio~(SIR).
In the max-SIR association, one does not have to consider the
distribution of distance to the serving BS.
On the other hand, the nearest BS association is the single-tier
homogeneous version of max-averaged-power association, where a user is
associated with the BS from which the user receives the maximum signal
power averaged over fading.
Afshang and Dhillon~\cite{AfshDhil18} also consider a model of
two-tier HetNets, where locations of users and small BSs are both
distributed according to PPCPs with the same parent process while
macro BSs are deployed according to an independent PPP.
In their model, a user can connect to any macro BSs but to the small
BSs with the same parent point.
For D2D networks, Afshang~\textit{et al.}~\cite{AfshDhilChon16},
Yi~\textit{et   al.}~\cite{YiLiuNall17} and Joshi and
Mallik~\cite{JoshMall18} consider the models, where user devices are
distributed according to a PPCP and a device communicates only with
another device in the same cluster.

We here consider the most fundamental setup of downlink cellular
networks, where single-tier BSs are deployed according to a PPCP.
Under the assumption of Rayleigh fading and the nearest BS
association, we derive a numerically computable form of coverage
probability.
To do this, we first derive the conditional coverage probability given
the parent process.
Since a PPCP is in the class of Cox (doubly stochastic Poisson)
processes (see, e.g., \cite{ChiuStoyKendMeck13}), it is conditionally
an inhomogeneous PPP provided the parent process.
Therefore, we can apply the discussion for PPP networks and then
arrive at the goal by unconditioning.
To validate the analysis, we compare the results of numerical
computations with those by Monte Carlo simulations.

\section{Poisson-Poisson cluster processes}\label{sec:PPCP}

A stationary PPCP on $\R^2$ is constructed by an independently marked
homogeneous PPP as follows.
Let $\Phi^{(p)} = \{X_i\}_{i\in\N}$ denote a homogeneous PPP on
$\R^2$, called a parent process, with intensity $\lambda_p$.
A mark $\Psi_i=\{Y_{i,j}\}_{j\in\N}$ of the point $X_i$ is a finite
(therefore inhomogeneous) PPP on $\R^2$, called a daughter process,
with intensity function $\lambda_d(x)$, $x\in\R^2$, satisfying
$\int_{\R^2}\lambda_d(x)\,\dd x = \alpha$; that is, the number of
daughter points per parent follows a Poisson distribution with mean
$\alpha$.
Then, a PPCP is given as $\Phi = \{Z_i\}_{i\in\N} =
\bigcup_{i\in\N}\{X_i + \Psi_i\}$, which is stationary with intensity
$\lambda_p\alpha$.
Throughout this paper, we focus on radially symmetric daughter
processes, so that $\lambda_d(x) = \alpha\,f_d(\|x\|)$ and $\Phi$ is
isotropic as well.
Two main examples of the PPCPs are the (modified) Thomas PP~(TPP) and
the Mat\'{e}rn cluster process~(MCP)~(see, e.g.,
\cite{ChiuStoyKendMeck13}).
When $f_d(s) = f_d^{\TPP}(s) =
\exp\{-s^2/(2\sigma^2)\}/(2\pi\sigma^2)$, $\sigma>0$, the PPCP is
called the TPP, where daughter points are independently and normally
distributed around each given parent point with covariance
matrix~$\sigma^2I$ ($I$ denotes the identity matrix).
On the other hand, when $f_d(s) = f_d^{\MCP}(s) =
\ind{[0,r_d]}(s)/(\pi{r_d}^2)$, $r_d>0$, the PPCP is called the MCP,
where daughter points are independently and uniformly distributed on
the ball of radius~$r_d$ centered at each given parent point.
PPCPs are a class of Cox PPs, so that, when the parent
process~$\Phi^{(p)} = \{X_i\}_{i\in\N}$ is provided, the PPCP~$\Phi$
is conditionally an inhomogeneous PPP with the shot-noise intensity
function;
\begin{equation}\label{eq:ConditionalIntensity}
  \lambda(y\mid\Phi^{(p)})
  = \sum_{i=1}^\infty\lambda_d(y-X_i)
  = \alpha\sum_{i=1}^\infty f_d(\|y - X_i\|),
  \quad y\in\R^2.
\end{equation}

For a stationary PP~$\Phi$ on $\R^2$, contact distance of $\Phi$
is defined as the distance from an arbitrary fixed reference point on
$\R^2$ to the nearest point of $\Phi$.
Here, due to the stationarity, we can choose the origin $o=(0,0)$ as
the reference point.
The conditional distribution of the contact distance given the parent
process is derived as follows.

\begin{lemma}\label{lem:ConditionalCDDistribution}
Let $\Phi$ denote a PPCP described above.
The conditional distribution function of contact distance of $\Phi$
provided the parent process~$\Phi^{(p)} = \{X_i\}_{i\in\N}$ is
given by
\begin{equation}\label{eq:ConditionalCDDistribution}
  F_{\mathrm{cd}}(r \mid \Phi^{(p)})
  = 1 - \prod_{i=1}^\infty  
          \exp\bigl\{
            - \alpha\,G(r \mid \|X_i\|)
          \bigr\},
\end{equation}
where
$G(r \mid s) = 2\int_0^r\int_0^\pi
  u\,f_d\bigr(\sqrt{u^2 + s^2 - 2\,u s\cos\phi}\bigr)\,
\dd\phi\,\dd u$.
Furthermore, the corresponding conditional density function is given
by
\begin{equation}\label{eq:ConditionalDensity}
  f_{\mathrm{cd}}(r \mid \Phi^{(p)})
  = \alpha
    \sum_{i=1}^\infty
       g(r \mid \|X_i\|)
    \prod_{i=1}^\infty
      \exp\bigl\{
        - \alpha\,G(r \mid \|X_i\|)
      \bigr\},
\end{equation}
where $g(r \mid s) = \partial G(r \mid s)/\partial r = 2\,r\int_0^\pi
f_d\bigr(\sqrt{r^2 + s^2 - 2\,r s\cos\phi}\bigr)\, \dd\phi$.
\end{lemma}

\begin{proof}
Let $b_o(r)$, $r>0$, denote the ball on $\R^2$ centered at the origin
with radius~$r$.
When the parent process $\Phi^{(p)} = \{X_i\}_{i\in\N}$ is provided,
$\Phi$ is (conditionally) an inhomogeneous PPP with the intensity
function given in \eqref{eq:ConditionalIntensity}.
Therefore, the conditional probability that $\Phi$ has no points in
$b_o(r)$ is given by
\[
  \Prb\bigl(
    \Phi\bigl(b_o(r)\bigr)=0 \bigm| \Phi^{(p)}
  \bigr)
  = \prod_{i=1}^\infty  
       \exp\Bigl\{
         -\alpha
          \int_{b_o(r)} f_d(\|y-X_i\|)\,\dd y
       \Bigr\}.
\]
Putting $y = (u\cos\phi, u\sin\phi)$ and $X_i=(X_{i,1},X_{i,2})$ in
the integral on the right-hand side above yields
\begin{align*}
  \int_{b_o(r)} f_d(\|y-X_i\|)\,\dd y
  &= \int_0^r\!\!\int_0^{2\pi}
       u\,
       f_d\Bigl(\sqrt{
         (u\cos\phi-X_{i,1})^2 + (u\sin\phi-X_{i,2})^2
       }\Bigr)\,
     \dd\phi\,\dd u
  \\
  &= G(r \mid \|X_i\|),
\end{align*}
which leads to \eqref{eq:ConditionalCDDistribution} since
$F_{\mathrm{cd}}(r \mid \Phi^{(p)}) = 1 - \Prb\bigl(
\Phi\bigl(b_o(r)\bigr)=0 \bigm| \Phi^{(p)} \bigr)$.
Differentiating \eqref{eq:ConditionalCDDistribution} with respect to
$r$ gives \eqref{eq:ConditionalDensity}.
\end{proof}

Note that $G(r \mid \cdot)$ and $g(r \mid \cdot)$ in
Lemma~\ref{lem:ConditionalCDDistribution} are, respectively, a
probability distribution function and the corresponding density
function with respect to $r\in[0,\infty)$ in the sense that
$\lim_{r\to\infty}G(r \mid \|x\|) = \int_0^\infty g(r \mid \|x\|)\,\dd
  r = \int_{\R^2} f_d(\|y-x\|)\,\dd y = 1$ for any $x\in\R^2$.
The distribution~$G(\cdot\mid s)$ gives the conditional distribution
of the distance to a daughter point from the origin provided that its
parent point is located at $x$ satisfying $\|x\|=s$.

\begin{example}[TPP]
For the TPP, applying $f_d^{\TPP}(s) =
\exp\{-s^2/(2\sigma^2)\}/(2\pi\sigma^2)$, the conditional distribution
$G(r \mid s)$ on the right-hand side of
\eqref{eq:ConditionalCDDistribution} reduces to
\begin{equation}\label{eq:G_TPP}
  G^{\TPP}(r \mid s)
  = \frac{1}{\sigma^2}
    \int_0^r
      u\,
      \exp\Bigl(-\frac{u^2 + s^2}{2\sigma^2}\Bigr)\,
      I_0\Bigl(\frac{u\,s}{\sigma^2}\Bigr)\,
    \dd u
  = 1 - Q_1\Bigl(\frac{s}{\sigma},\frac{r}{\sigma}\Bigr),
\end{equation}
where $I_0$ denotes the modified Bessel function of the first kind
with order zero;
$I_0(z)=\pi^{-1}\int_0^\pi\ee^{z\cos\phi}\,\dd\phi$, and $Q_1$
denotes the first-order Marcum $Q$-function defined as (see, e.g.,
\cite{ProaSale07})
\[
  Q_1(a,b)
  = \int_b^\infty
      x\,\exp\Bigl(-\frac{x^2 + a^2}{2}\Bigr)\,I_0(a\,x)\,\dd x.
\]
Therefore, differentiating \eqref{eq:G_TPP} gives the density function;
\begin{equation}\label{eq:g_TPP}
  g^{\TPP}(r\mid s)
  = \frac{1}{\sigma}\,
    q\Bigl(\frac{s}{\sigma}, \frac{r}{\sigma}\Bigr),
\end{equation}
where $q(a,b) = - \partial Q_1(a,b)/\partial b = b\,\exp\bigl(-(a^2 +
b^2)/2\bigr)\,I_0(ab)$.
\end{example}

\begin{example}[MCP]
For the MCP, since $f_d^{\MCP}(s) = \ind{[0,r_d]}(s)/(\pi{r_d}^2)$,
the conditional distribution $G(r \mid s)$ in
\eqref{eq:ConditionalCDDistribution} reduces to
\begin{align}\label{eq:G_MCP}
  G^{\MCP}(r \mid s)
  &= \frac{2}{\pi\,{r_d}^2}
     \int_0^r\!\!\int_0^\pi
       u\,\bsym{1}\Bigl\{
            \cos\phi \ge \frac{u^2 + s^2 - {r_d}^2}{2\,u s}
          \Bigr\}\,
     \dd\phi\,\dd u
  \nonumber\\
  &= \frac{1}{{r_d}^2}\,
     \bigg\{
       \bigl[r\wedge(r_d-s)^+\bigr]^2
       + \frac{2}{\pi}
         \int_{r\wedge|r_d-s|}^{r\wedge(r_d+s)}
           u\,\arccos\Bigl(\frac{u^2 + s^2 - {r_d}^2}
                                {2\,u s}\Bigr)\,
         \dd u
     \biggr\},
\end{align} 
where $x^+=\max(x,0)$, $x\wedge y=\min(x,y)$, and we use
$\int_0^\pi\bsym{1}\{\cos\phi\ge x\}\,\dd\phi =
\pi\,\ind{(-\infty,-1]}(x) + \arccos x\,\ind{(-1,1]}(x)$ in the second
equality.
Hence, the density function is given as
\begin{equation}\label{eq:g_MCP}
  g^{\MCP}(r\mid s)
  = \frac{2\,r}{{r_d}^2}\,
    \biggl\{
      \ind{[0,(r_d - s)^+]}(r)
      + \frac{1}{\pi}\,
        \arccos\Bigl(
          \frac{r^2 + s^2 - {r_d}^2}{2\,r s}
        \Bigr)\,
        \ind{[|r_d-s|, r_d+s]}(r)
    \biggr\}.
\end{equation}
\end{example}

Note that \eqref{eq:g_TPP} has the same form as (3) in
\cite{AfshSahaDhil17} and that \eqref{eq:g_MCP} does so as the couple
of (2) and (3) in \cite{AfshSahaDhil17b}.
We can obtain the same results as in
\cite{AfshSahaDhil17,AfshSahaDhil17b} by plugging \eqref{eq:G_TPP} or
\eqref{eq:G_MCP} into \eqref{eq:ConditionalCDDistribution} and then
unconditioning it on $\Phi^{(p)}$ with the use of the probability
generating functional~(PGFL) for PPPs~(see, e.g.,
\cite{ChiuStoyKendMeck13}).
In other words, we have a unified form of contact distance
distributions for PPCPs as
\[
  F_{\mathrm{cd}}(r)
  = 1 - \exp\biggl\{
          - 2\,\pi\lambda_p
            \int_0^\infty
              \bigl[
                1 - \exp\bigl\{-\alpha\,G(r\mid s)\bigr\}
              \bigr]\,s\,
            \dd s
        \biggr\},
  \quad r\ge 0.
\]

\section{Coverage probability for a downlink cellular network}

We consider a fundamental model of single-tier homogeneous downlink
cellular networks.
Let $\Phi = \{Z_i\}_{i\in\N}$ denote a stationary PP on $\R^2$
representing locations of BSs, where the order of the points is
arbitrary but $Z_1$ is the nearest from the origin; that is,
$\|Z_1\|<\|Z_i\|$ for $i\ge2$.
All the BSs transmit signals at the same power level and each user is
associated with the nearest BS.
Due to the stationarity and homogeneity, we can focus on a typical
user located at the origin.
For each $i\in\N$, let $H_i$ denote a nonnegative random variable
representing a fading effect on a signal from the BS at $Z_i$ to the
typical user, where we assume Rayleigh fading and $H_i$, $i\in\N$, are
mutually independent and exponentially distributed, as well as
independent of $\Phi$.
We assume $\Exp H_1=1$ without loss of generality and ignore
shadowing.
The path-loss function representing signal attenuation with distance
is given by $\ell(r)$, $r>0$, which satisfies $\int_\epsilon^\infty
r\,\ell(r)\,\dd r<\infty$ for any $\epsilon>0$.
With this setup, the SIR for the typical user is defined as
\begin{equation}\label{eq:SIR}
  \SIR_o = \frac{H_1\,\ell(\|Z_1\|)}
                {\sum_{i=2}^\infty H_i\,\ell(\|Z_i\|)}.
\end{equation}
Since $Z_1$ is the nearest point of $\Phi$ from the origin, $\|Z_1\|$
gives the contact distance of $\Phi$.
Our interest is in the coverage probability $\Prb(\SIR_o > \theta)$ for
$\theta > 0$ when the PP~$\Phi$ is given as a PPCP.

\begin{theorem}\label{thm:main}
For the downlink cellular network model described above, when the
PP~$\Phi$ is a stationary PPCP given in Section~\ref{sec:PPCP}, we
have
\begin{equation}\label{eq:Coverage}
  \Prb(\SIR_o > \theta)
  = \alpha\int_0^\infty T(r,\theta)\,M(r,\theta)\,\dd r,
\end{equation}
where
\begin{align*}
  T(r,\theta)
  &= 2\,\pi\lambda_p
     \int_0^\infty
       g(r \mid s)\,C(r, s, \theta)\,s\,
     \dd s,
  \\  
  M(r,\theta)
  &= \exp\biggl\{
       - 2\,\pi\lambda_p
         \int_0^\infty
           \bigl[1 - C(r,s,\theta)\bigr]\,s\,
         \dd s
     \biggr\},
\end{align*}  
with $g(\cdot\mid\cdot)$ given in
Lemma~\ref{lem:ConditionalCDDistribution} and
\begin{equation}\label{eq:C}
  C(r,s,\theta)
  = \exp\biggl\{
      - \alpha\,
        \biggl[
          1 - \int_r^\infty
                \Bigl(
                  1 + \theta\,\frac{\ell(u)}{\ell(r)}
                \Bigr)^{-1}\,
                g(u \mid s)\,
              \dd u
        \biggr]
    \biggr\}.
\end{equation}
\end{theorem}

\begin{proof}
Since $\Phi$ is conditionally an inhomogeneous PPP provided $\Phi^{(p)}
= \{X_i\}_{i\in\N}$, we can follow the standard discussion for the PPP
network with Rayleigh fading~(see, e.g.,
\cite{AndrBaccGant11} for the homogeneous PPP case); that is,
application of \eqref{eq:ConditionalIntensity} and
\eqref{eq:ConditionalDensity} to \eqref{eq:SIR} leads to
\begin{align}\label{eq:ConditionalCoveragePrf1}
  \Prb(\SIR_o > \theta \mid \Phi^{(p)})
  &= \Exp\biggl[
       \prod_{i=2}^\infty
         \Bigl(
           1 + \theta\,\frac{\ell(\|Z_i\|)}{\ell(\|Z_1\|)}
         \Bigr)^{-1}
     \biggm| \Phi^{(p)}\biggr]
  \nonumber\\
  &= \int_0^\infty
       f_{\mathrm{cd}}(r \mid \Phi^{(p)})\,
       \exp\biggl\{
         - \int_{\|y\|>r}
             \biggl[
               1 - \Bigl(
                     1 + \theta\,\frac{\ell(\|y\|)}{\ell(r)}
                   \Bigr)^{-1}
             \biggr]\,
             \lambda(y\mid\Phi^{(p)})\,
           \dd y
       \biggr\}\,
     \dd r,
\end{align}
where we use the distribution function and the Laplace transform of
exponential random variables in the first equality, and the PGFL for
(inhomogeneous) PPPs in the second equality.
By \eqref{eq:ConditionalIntensity} and $g(\cdot\mid\cdot)$ in
Lemma~\ref{lem:ConditionalCDDistribution}, the integral inside the
exponential function above is equal to the sum over $i\in\N$ of
\begin{align*}
  \alpha
  \int_{\|y\|>r}
    \Bigl[
      1 - \Bigl(
            1 + \theta\,\frac{\ell(\|y\|)}{\ell(r)}
          \Bigr)^{-1}
    \Bigr]\,
    f_d(\|y - X_i\|)\,
  \dd y  
  &= \alpha
     \int_r^\infty
       \Bigl[
         1 - \Bigl(
               1 + \theta\,\frac{\ell(u)}{\ell(r)}
             \Bigr)^{-1}
       \Bigr]\,
       g(u \mid \|X_i\|)\,
     \dd u
  \\
  &= \alpha\biggl[
       \Bar{G}(r\mid\|X_i\|)
       - \int_r^\infty
           \Bigl(
             1 + \theta\,\frac{\ell(u)}{\ell(r)}
           \Bigr)^{-1}\,
           g(u \mid \|X_i\|)\,
         \dd u
     \biggr],
\end{align*}  
where $\Bar{G}(r\mid s) = 1 - G(r\mid s)$ and the same discussion as
in the proof of Lemma~\ref{lem:ConditionalCDDistribution} is used.
Therefore, plugging  this into \eqref{eq:ConditionalCoveragePrf1} and
using \eqref{eq:ConditionalDensity}, we reduce
\eqref{eq:ConditionalCoveragePrf1} to
\[
  \Prb(\SIR_o > \theta \mid \Phi^{(p)})
  = \alpha
    \int_0^\infty
      \sum_{i=1}^\infty
        g(r \mid \|X_i\|)
      \prod_{i=1}^\infty
        C(r, \|X_i\|, \theta)\,
    \dd r,
\]
where $C$ is given in \eqref{eq:C}.
Hence, unconditioning on $\Phi^{(p)}$, we have
\begin{align}\label{eq:ConditionalCoveragePrf2}
  \Prb(\SIR_o > \theta)
  &= \alpha
     \int_0^\infty
       \Exp\biggl[
         \sum_{i=1}^\infty
           g(r \mid \|X_i\|)
         \prod_{i=1}^\infty
           C(r, \|X_i\|, \theta)
       \biggr]\,
     \dd r
  \nonumber\\
  &= \alpha\,\lambda_p
     \int_0^\infty\!\!
       \int_{\R^2}
         g(r \mid \|x\|)\,C(r, \|x\|, \theta)\,
       \dd x\,
       \exp\biggl\{
         -\lambda_p
          \int_{\R^2}
            \bigl[1 - C(r, \|x\|, \theta)\bigr]\,
          \dd x
       \biggr\}\,
     \dd r,
\end{align}
where we apply the Campbell-Mecke formula (see, e.g.,
\cite{ChiuStoyKendMeck13}) by regarding $\prod_{j=1,j\ne i}^\infty
C(r, \|X_j\|, \theta)$ as a mark of the point~$X_i$, and then use the
PGFL for homogeneous PPPs in the second equality (this type of
transform is also used in \cite{SchiToumHaenCrisBranBett16}).
It is immediate to see that the right-hand side of
\eqref{eq:ConditionalCoveragePrf2} is equal to that of
\eqref{eq:Coverage}.
Actually, we have to confirm whether the PGFL is applicable in
\eqref{eq:ConditionalCoveragePrf2} and this is done by showing
$\lambda_p\int_{\R^2}\bigl|\log C(r, \|x\|,\theta)\bigr|\,\dd x <
\infty$ (see, e.g., \cite[pp. 59--60]{DaleVere08}).
By \eqref{eq:C}, noting that $g(\cdot\mid s)$ is a probability density
function for any $s>0$, we have
\begin{align}\label{eq:appdx1}
  \int_{\R^2}
    \bigl|\log C(r, \|x\|,\theta)\bigr|\,
  \dd x
  &= 2\,\pi\alpha
    \int_0^\infty
      \biggl[
        1 - \int_r^\infty
              \Bigl(
                1 + \theta\,\frac{\ell(u)}{\ell(r)}
              \Bigr)^{-1}\,
              g(u \mid s)\,
            \dd u
      \biggr]\,
    s\,\dd s
  \nonumber\\
  &= 2\,\pi\alpha
    \int_0^\infty
      \biggl[
        G(r\mid s)
        + \theta
          \int_r^\infty
            \frac{\ell(u)}{\theta\,\ell(u) + \ell(r)}\,
            g(u \mid s)\,
          \dd u
     \biggr]\,
    s\,\dd s.
\end{align}
For the first term in the integrand above, the symmetry of $s\,g(r\mid
s) = r\,g(s\mid r)$ (see Lemma~\ref{lem:ConditionalCDDistribution})
implies
\[
  \int_0^\infty G(r\mid s)\,s\,\dd s = \int_0^r u\,\dd u =
  \frac{r^2}{2} < \infty,
\]
where we use $\int_0^\infty g(s\mid r)\,\dd s = 1$.
For the second term in the integrand of \eqref{eq:appdx1}, we have
similarly
\begin{align*}
  \int_0^\infty\!\!\int_r^\infty
    \frac{\ell(u)}
         {\theta\,\ell(u) + \ell(r)}\,
    g(u \mid s)\,
  \dd u\,s\,\dd s
  &= \int_r^\infty
       \frac{u\,\ell(u)}
            {\theta\,\ell(u) + \ell(r)}\,
     \dd u
  \\
  &\le \frac{1}{\ell(r)}\int_r^\infty u\,\ell(u)\,\dd u < \infty,
\end{align*}      
where the last inequality follows from the assumption on the path-loss
function.
\end{proof}

\section{Numerical experiments}

Figure~\ref{fig:TPP} and \ref{fig:MCP} display the comparison results
of numerical computations based on our analysis and Monte Carlo
simulations.
Throughout the experiments, we fix the path-loss function as $\ell(r)
= r^{-4}$, $r>0$, and the parameters $\lambda_p = 0.1/\pi$, $\alpha =
10$.
In both the TPP and the MCP, three cases of $\Exp[\|Y_{i,j}\|^2] =
0.6$, $1.4$, and $3.0$ are computed (that is, $\sigma^2=0.3$, $0.7$,
and $1.5$ in the TPP, and ${r_d}^2 = 1.2$, $2.8$, and $6.0$ in the
MCP).
In each simulation run, samples of parent points are put on the disk
with radius~$100$ and daughter points are scattered around the parent
points.
Then, the estimated coverage probability is obtained by average taken
over 20,000 independent copies.
The agreement between the theoretical and simulation results supports
the validity of our analysis.

\begin{figure}
\centering%
\includegraphics[width=.7\linewidth]{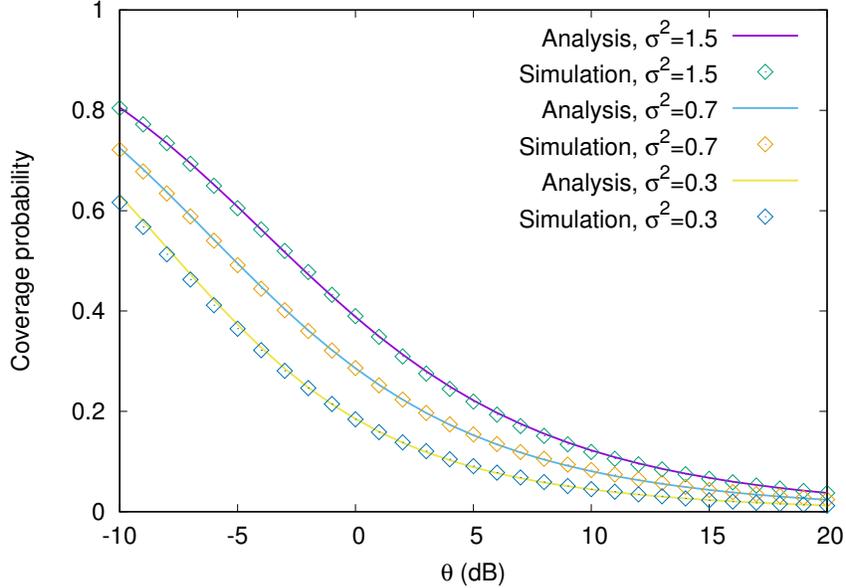}
\caption{Coverage probability in the TPP network with
  $\lambda_p=0.1/\pi$, $\alpha=10$.}
\label{fig:TPP}
\end{figure}

\begin{figure}
\centering%
\includegraphics[width=.7\linewidth]{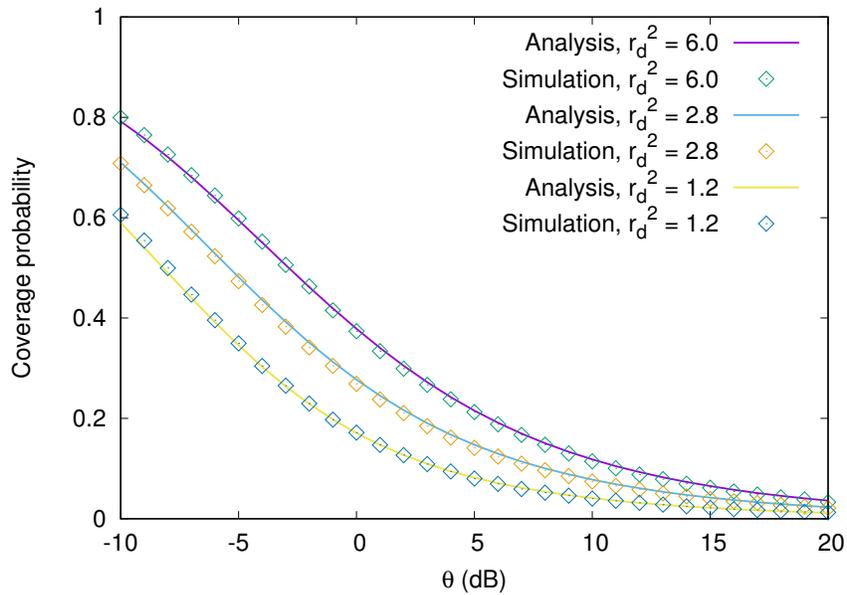}
\caption{Coverage probability in the MCP network with
  $\lambda_p=0.1/\pi$, $\alpha=10$.}
\label{fig:MCP}
\end{figure}

We should notice that the actual numerical computation of the coverage
probability using \eqref{eq:Coverage} is not so easy.
In particular, the integral inside function~$M$ in
Theorem~\ref{thm:main} hardly converges in a numerical sense (though
the finite existence is ensured in the proof of
Theorem~\ref{thm:main}) and we should take a truncation technique
carefully.

\section{Conclusion}

We have considered a spatial downlink cellular network model with BSs
deployed according to a PPCP and, within the most fundamental setup
such as single-tier, Rayleigh fading and nearest BS association, we
have derived the coverage probability in a numerically computable
form.
This work does not only fill in a hole of the literature but also is
expected to play a role of a building block for analysis of, for
example, HetNets with a tier consisting of open access small cell BSs.

\section*{Acknowledgments}

Thanks are due to one of the anonymous reviewers of the author's
another paper~\cite{TakaChenKobaMiyo17}, whose comment asking the
existence of analytical results for coverage probability in TPP
networks motivated the author.

The author's work was supported by the Japan Society for the Promotion
of Science (JSPS) Grant-in-Aid for Scientific Research (C) 16K00030.
  
\bibliographystyle{IEEETran}
\bibliography{../references}

\begin{thebibliography}{10}
\providecommand{\url}[1]{#1}
\csname url@samestyle\endcsname
\providecommand{\newblock}{\relax}
\providecommand{\bibinfo}[2]{#2}
\providecommand{\BIBentrySTDinterwordspacing}{\spaceskip=0pt\relax}
\providecommand{\BIBentryALTinterwordstretchfactor}{4}
\providecommand{\BIBentryALTinterwordspacing}{\spaceskip=\fontdimen2\font plus
\BIBentryALTinterwordstretchfactor\fontdimen3\font minus
  \fontdimen4\font\relax}
\providecommand{\BIBforeignlanguage}[2]{{%
\expandafter\ifx\csname l@#1\endcsname\relax
\typeout{** WARNING: IEEEtran.bst: No hyphenation pattern has been}%
\typeout{** loaded for the language `#1'. Using the pattern for}%
\typeout{** the default language instead.}%
\else
\language=\csname l@#1\endcsname
\fi
#2}}
\providecommand{\BIBdecl}{\relax}
\BIBdecl

\bibitem{BlasYoge14}
B.~B{\l}aszczyszyn and D.~Yogeshwaran, ``Clustering comparison of point
  processes with applications to random geometric models,'' in \emph{Stochastic
  Geometry, Spatial Statistics and Random Fields: Models and Algorithms},
  V.~Schmidt, Ed.\hskip 1em plus 0.5em minus 0.4em\relax Springer, 2014, pp.
  31--71.

\bibitem{SuryMollFett15}
V.~Suryaprakash, J.~M{\o}ller, and G.~Fettweis, ``On the modeling and analysis
  of heterogeneous radio access networks using a {Poisson} cluster process,''
  \emph{IEEE Transactions on Wireless Communications}, vol.~14, pp. 1035--1047,
  2015.

\bibitem{DengZhouHaen15}
N.~Deng, W.~Zhou, and M.~Haenggi, ``Heterogeneous cellular network models with
  dependence,'' \emph{IEEE Journal on Selected Areas in Communications},
  vol.~33, pp. 2167--2181, 2015.

\bibitem{ChunHasnGhra15}
Y.~J. Chun, M.~O. Hasna, and A.~Ghrayeb, ``Modeling heterogeneous cellular
  networks interference using {Poisson} cluster processes,'' \emph{IEEE Journal
  on Selected Areas in Communications}, vol.~33, pp. 2182--2195, 2015.

\bibitem{SahaAfshDhil17}
C.~Saha, M.~Afshang, and H.~S. Dhillon, ``{3GPP}-inspired {HetNet} model using
  {Poisson} cluster process: {Sum-product} functionals and downlink coverage,''
  \emph{IEEE Transactions on Communications}, 2017, in press.

\bibitem{AfshDhil18}
M.~Afshang and H.~S. Dhillon, ``{Poisson} clustered process based analysis of
  {HetNets} with correlated user and base station locations,'' \emph{IEEE
  Transactions on Wireless Communications}, 2018, in press.

\bibitem{AfshDhilChon16}
M.~Afshang, H.~S. Dhillon, and P.~H.~J. Chong, ``Modeling and performance
  analysis of clustered device-to-device networks,'' \emph{IEEE Transactions on
  Wireless Communications}, vol.~15, pp. 4957--4972, 2016.

\bibitem{YiLiuNall17}
W.~Yi, Y.~Liu, and A.~Nallanathan, ``Modeling and analysis of {D2D}
  millimeter-wave networks with {Poisson} cluster processes,'' \emph{IEEE
  Transactions on Communications}, vol.~65, pp. 5574--5588, 2017.

\bibitem{JoshMall18}
S.~Joshi and R.~K. Mallik, ``Coverage and interference in {D2D} with {Poisson}
  cluster process,'' \emph{IEEE Communications Letters}, 2018, in press.

\bibitem{AfshSahaDhil17}
M.~Afshang, C.~Saha, and H.~S. Dhillon, ``Nearest-neighbor and contact distance
  distributions for {Thomas} cluster process,'' \emph{IEEE Wireless
  Communications Letters}, vol.~6, pp. 130--133, 2017.

\bibitem{AfshSahaDhil17b}
------, ``Nearest-neighbor and contact distance distributions for {Mat\'{e}rn}
  cluster process,'' \emph{IEEE Communications Letters}, vol.~21, pp.
  2686--2689, 2017.

\bibitem{BlasHaenKeelMukh18}
B.~B{\l}aszczyszyn, M.~Haenggi, P.~Keeler, and S.~Mukherjee, \emph{Stochastic
  Geometry Analysis of Cellular Networks}.\hskip 1em plus 0.5em minus
  0.4em\relax Cambridge University Press, 2018.

\bibitem{ChiuStoyKendMeck13}
S.~N. Chiu, D.~Stoyan, W.~S. Kendall, and J.~Mecke, \emph{Stochastic Geometry
  and its Applications}, 3rd~ed.\hskip 1em plus 0.5em minus 0.4em\relax Wiley,
  2013.

\bibitem{ProaSale07}
J.~Proakis and M.~Salehi, \emph{Digital Communications}, 5th~ed.\hskip 1em plus
  0.5em minus 0.4em\relax McGraw-Hill, 2007.

\bibitem{AndrBaccGant11}
J.~G. Andrews, F.~Baccelli, and R.~K. Ganti, ``A tractable approach to coverage
  and rate in cellular networks,'' \emph{IEEE Transactions on Communications},
  vol.~59, pp. 3122--3134, 2011.

\bibitem{SchiToumHaenCrisBranBett16}
U.~Schilcher, S.~Toumpis, M.~Haenggi, A.~Crismani, G.~Brandner, and
  C.~Bettstetter, ``Interference functionals in {Poisson} networks,''
  \emph{IEEE Transactions on Information Theory}, vol.~62, pp. 370--383, 2016.

\bibitem{DaleVere08}
D.~J. Daley and D.~Vere-Jones, \emph{An Introduction to the Theory of Point
  Processes: Volume II: General Theory and Structure}, 2nd~ed.\hskip 1em plus
  0.5em minus 0.4em\relax Springer, 2008.

\bibitem{TakaChenKobaMiyo17}
Y.~Takahashi, Y.~Chen, T.~Kobayashi, and N.~Miyoshi, ``Simple and fast
  {PPP}-based approximation of {SIR} distributions in downlink cellular
  networks,'' 2017, submitted.

\end{thebibliography}


\end{document}